\documentclass{article}[11pt]

\usepackage{amsmath} 
\usepackage{amssymb}
\usepackage{tikz}
\usetikzlibrary{arrows,shapes,shapes.misc,shadows,shadings,decorations.markings,automata}

\tikzset{cross/.style={cross out, draw=blue, minimum size=2*(#1-\pgflinewidth), inner sep=0pt, outer sep=0pt},
cross/.default={1.5pt}}

\usepackage{verbatim}
\usepackage{changepage}
\usepackage{pgfplots}

\usepackage[hmarginratio=1:1,top=32mm]{geometry} 

\usepackage{enumitem} 

\usepackage{titling} 
\usepackage{amsmath}
\usepackage{amsthm}
\usepackage{hyperref} 

\usepackage{caption}
\usepackage[labelformat=simple]{subcaption}
\usepackage{graphicx}
\usepackage{xcolor}
\definecolor{red}{RGB}{146,0,0}
\definecolor{blue}{RGB}{0,109,219}
\definecolor{green}{RGB}{36,255,36}

\DeclareMathOperator{\sdpval}{SDPVal}
\DeclareMathOperator{\optval}{OptVal}
\DeclareMathOperator{\Cov}{Cov}
\DeclareMathOperator{\opt}{Opt}
\DeclareMathOperator{\val}{Val}
\DeclareMathOperator{\rndval}{RndVal}

\newcommand*\samethanks[1][\value{footnote}]{\footnotemark[#1]}

\pretitle{\begin{center}\Huge\bfseries} 
\posttitle{\end{center}} 
\title{Global Cardinality Constraints Make Approximating Some Max-2-CSPs Harder} 
\author{%
  {\textsc{Per Austrin} \thanks{Research supported by the Approximability and Proof Complexity project funded by the Knut and Alice Wallenberg Foundation.} } \\[1ex]  
\normalsize School of Electrical Engineering and Computer Science
\\ \normalsize KTH Royal Institute of Technology\\ 
\normalsize \href{mailto:austrin@kth.se}{austrin@kth.se} 
\and 
{\textsc{Aleksa Stankovi\'c} \samethanks} \\[1ex] 
\normalsize Department of Mathematics \\ 
\normalsize KTH Royal Institute of Technology\\ 
\normalsize \href{mailto:aleksas@kth.se}{aleksas@kth.se} 
}
\date{\today} 
\renewcommand{\maketitlehookd}{%
\begin{abstract}
  Assuming the Unique Games Conjecture, we show that existing
  approximation algorithms for some Boolean Max-$2$-CSPs with
  cardinality constraints are optimal.  In particular, we prove that
  Max-Cut with cardinality constraints is UG-hard to approximate
  within $\approx 0.858$, and that Max-$2$-Sat with cardinality
  constraints is UG-hard to approximate within $\approx 0.929$.  In
  both cases, the previous best hardness results were the same as
  the hardness of the corresponding unconstrained Max-$2$-CSP
  ($\approx 0.878$ for Max-Cut, and $\approx 0.940$ for Max-$2$-Sat).

  The hardness obtained for Max-$2$-Sat applies to monotone Max-$2$-Sat instances,
  meaning that we also obtain tight inapproximability for
  the Max-$k$-Vertex-Cover problem.
\noindent 
\end{abstract}
}

\hypersetup{
  colorlinks,
  citecolor=black,
  linkcolor=black,
  urlcolor=blue}

\theoremstyle{plain}
\newtheorem{thm}{Theorem}[section] 
\newtheorem{lemma}[thm]{Lemma} 
\newtheorem{conj}[thm]{Conjecture} 
\newtheorem{claim}[thm]{Claim}
\theoremstyle{definition}
\newtheorem{defn}[thm]{Definition} 

\begin{document}

\maketitle


\section{Introduction}
Constraint satisfaction problems (CSPs) are one of the most fundamental objects studied in complexity theory.
An instance of a CSP has a set of variables taking values over a certain domain and a set of constraints on tuples of these variables as an input.
Probably the best known CSP is 3-Sat, in which the constraints are clauses, each clause is a disjunction of at most three literals, and each literal is either a variable or negation of a variable.
In the satisfiability version of CSP problems, we are interested whether there is an assignment to the variables which satisfies all the constraints. 
Hardness of deciding satisfiability of CSPs is well understood, due to the dichotomy theorem \cite{Schaefer:1978:CSP:800133.804350} of Schaefer which shows that each CSP with variables taking values in a Boolean domain is either in P or NP-complete, and due to the more recent results of Bulatov \cite{DBLP:conf/focs/Bulatov17} and Zhuk \cite{DBLP:conf/focs/Zhuk17} which settle this question on general domains.
\par 
Another well-studied version is the Max-CSP, which is the optimization version in which we are interested in maximizing the number of constraints satisfied.  This type of problem is NP-hard in most cases and we typically settle with finding a good estimate of the optimal solution, for which we rely on approximation algorithms.
A common example of a constraint satisfaction problem in this setting is Max-Cut, in which the input consists of a graph $G$, and the goal is to partition the vertices into two sets such that the number of edges between the two parts is maximized. 
Approximability of Max-CSPs has been a major research topic which inspired many influential breakthroughs. One of the first surprising results was an algorithm of Goemans and Williamson \cite{DBLP:conf/stoc/GoemansW94}, which uses semidefinite programming (SDP) to approximate the optimal solution to within a constant of $\alpha_{GW} \approx 0.878$.  The SDP approach is also useful in approximating many other well known Max-CSPs, such as Max-3-Sat \cite{DBLP:conf/focs/KarloffZ97} within a constant of $7/8$ and Max-2-Sat \cite{DBLP:conf/ipco/LewinLZ02} within $\alpha_{LLZ} \approx 0.9401$. 
\par 
On the hardness of approximation side, the first NP-hardness results are based on the celebrated PCP theorem \cite{DBLP:conf/focs/AroraLMSS92,DBLP:conf/focs/AroraS92} which provided a strong starting point for studying inapproximability. For example, a direct corollary of the PCP theorem shows that the Max-$3$-Sat problem cannot be approximated within $1-\delta$ for some universal constant $\delta>0$. By using the PCP theorem and parallel repetition \cite{DBLP:journals/siamcomp/Raz98} as a starting point, H\aa stad \cite{DBLP:journals/jacm/Hastad01} proved optimal inapproximability for Max-$3$-Sat by showing that it cannot be approximated better than $7/8+\epsilon$ for any $\epsilon > 0$. 
\par
However, despite further works relying on the similar techniques which improved our understanding of inapproximability for several additional CSPs, the progress on closing the gap between the best algorithm and the best hardness was at a standstill for some fundamental problems such as Max-Cut, until the Unique Games Conjecture (UGC) was introduced by Khot \cite{DBLP:conf/stoc/Khot02a}. In particular, by assuming the UGC, optimality of the $\alpha_{GW}$-approximation algorithm for Max-Cut and the $\alpha_{LLZ}$-approximation algorithm for Max-$2$-Sat was shown in \cite{DBLP:journals/siamcomp/KhotKMO07,MR2630040} and \cite{DBLP:conf/stoc/Austrin07}, respectively. The strength of semidefinite programming for approximating Max-CSPs was corroborated in a breakthrough result of Raghavendra \cite{DBLP:conf/stoc/Raghavendra08}, which showed that assuming the UGC, a certain SDP relaxation achieves optimal approximation ratios for all Max-CSPs.
\par 
Locality of the constraints was of crucial importance in studying CSPs and Max-CSPs since their inception. Therefore, it is not a surprise that typical techniques fail when we work with CSPs for which feasible assignments need to satisfy some additional global constraints, and these problems almost always become harder. For example, while the satisfiability of a $2$-Sat instance can be checked by a straightforward algorithm, Guruswami and Lee recently showed \cite{DBLP:journals/mst/GuruswamiL16} that when the satisfying assignment needs to have exactly half of its variables set to true, this problem becomes NP-hard. Hardness of deciding satisfiability of CSPs in which we prescribe how many variables are assigned to certain values is well understood due to the dichotomy theorem of Bulatov and Marx \cite{DBLP:journals/corr/abs-1010-0201}, which shows that these problems are either NP-hard or in P, and gives a simple classification. Another type of global constraint is studied by Brakensiek et al.~\cite{DBLP:journals/eccc/BrakensiekGG19}, who consider hardness of deciding CSPs in presence of modular constraints, which restrict cardinality of values in an assignment modulo a natural number $M$.
\par 
In this paper we are interested in optimization variants of CSPs with global cardinality constraints, i.e., constraints which specify the number of occurrences of each value from the domain in the assignment. We refer to these problems as CC-Max-CSPs. It is not hard to see that these problems are at least as hard to approximate as their unconstrained counterparts. CC-Max-CSPs have been actively studied in the past. For example the Max-Bisection problem, i.e., Max-Cut where the two partitions need to be of the same size, has been of a particular interest, with a series of papers \cite{DBLP:journals/algorithmica/FriezeJ97}, \cite{DBLP:journals/mp/Ye01},\cite{DBLP:journals/rsa/HalperinZ02}, \cite{DBLP:journals/jal/FeigeL06}, \cite{DBLP:conf/soda/RaghavendraT12} obtaining improved approximation algorithms, until the most recent result which achieves an approximation ratio of $0.8776$ \cite{DBLP:conf/soda/AustrinBG13}, which is only $\approx 10^{-3}$ below the UG-hardness bound $\alpha_{GW}$. The state-of-the-art algorithm \cite{DBLP:conf/soda/RaghavendraT12} for the more general CC-Max-Cut problem achieves an approximation ratio of ${\alpha}^{cc}_{cut} \approx 0.858$.  Another related CC-Max-CSP actively studied is CC-Max-2-Sat, and its monotone variant (a version in which negated literals are not allowed) Max-$k$-VC\footnote{Max-$k$-VC is an abbreviation for maximum $k$ vertex cover, in which we are given a graph and the task is to select a subset of $k$ vertices covering as many of the edges as possible.}. The best algorithm \cite{DBLP:conf/soda/RaghavendraT12} up to date for general CC-Max-2-Sat achieves an approximation ratio of ${\alpha}^{cc}_{2sat}$, where ${\alpha}^{cc}_{2sat} \approx 0.929$, which improved on a series of increasingly stronger algorithms presented in \cite{DBLP:journals/algorithmica/Sviridenko01}, \cite{DBLP:conf/isaac/BlaserM02}, and \cite{DBLP:conf/esa/Hofmeister03}. Manurangsi \cite{DBLP:conf/soda/Manurangsi19} showed that it is UG-hard to approximate Max-$k$-VC within a factor $\alpha_{AKS} \approx 0.944$ (note that this is slightly larger than the hardness of $\alpha_{LLZ} \approx 0.940$ for general Max-$2$-Sat).
\par
Yet another well-studied CC-Max-CSP is the \emph{Densest $k$-Subgraph} (Max-$k$-DS) problem, in which we are given a graph and the objective is to find a maximally dense induced subgraph on $k$ vertices.  Analogously to the Max-$k$-VC problem, Max-$k$-DS can be viewed as the monotone CC-Max-$2$-And problem.  Max-$k$-DS is qualitatively very different from the previously discussed problems.  It is not known to be approximable within a constant factor, and is in fact known to be hard to approximate to within almost polynomial factors assuming the Exponential Time Hypothesis \cite{DBLP:conf/stoc/Manurangsi17}, or to within any constant factor assuming the Small-Set Expansion Hypothesis \cite{DBLP:conf/stoc/RaghavendraS10}.
\par
Obtaining tight approximability results for CC-Max-CSPs presents an important research topic. 
Qualitatively, it is also interesting to determine whether adding a cardinality constraint to a non-trivial Max-CSP makes approximation strictly harder. For example, we know that CC-Max-$2$-Sat is as hard as Max-$2$-Sat, but it is still conceivable that they are equally hard.  In particular, it would be interesting to answer the following question:
\par
\begin{quote}
\begin{center}
 ``Can CC-Max-2-Sat be approximated within $\alpha_{LLZ}$?''
\end{center}
\end{quote}
\par 
So far the only result in this direction comes from \cite{DBLP:conf/soda/AustrinBG13} which shows that the ``bisection version'' (where the cardinality constraint is that exactly half of the variables must be set to true) of CC-Max-2-Sat can be approximated within $\alpha_{LLZ}$. However, the approach taken in that algorithm does not immediately extend to handle general cardinality constraints.   A similar question arises for the CC-Max-Cut problem, but here even the basic Max-Bisection problem is not known to be approximable within the Max-Cut constant $\alpha_{GW} \approx 0.878$.   As far as we are aware, prior to this paper, the only examples of cardinality-constrained Max-CSPs being harder than their unconstrained counterparts were examples where the unconstrained version is easy (e.g.~unconstrained Max-$k$-VC is monotone Max-$2$-Sat, and unconstrained Max-$k$-DS is monotone Max-$2$-And, which are both trivial).
\par
\paragraph{Our results}
In this paper, we answer the above question negatively, by giving improved UG-hardness results for CC-Max-Cut and Max-$k$-VC.
\begin{thm}\label{basic_max-cut-sol}
  For every $\varepsilon > 0$, CC-Max-Cut is UG-hard to approximate within $\beta^{cc}_{cut}+ \varepsilon$, where $\beta^{cc}_{cut} \approx 0.858$. 
\end{thm}

\begin{thm} \label{basic_max-k-vc}
  For every $\varepsilon > 0$, Max-$k$-VC is UG-hard to approximate within $\beta^{cc}_{vc} +\varepsilon$, where $\beta^{cc}_{vc} \approx 0.929$. 
\end{thm}
Note that since CC-Max-Cut and Max-$k$-VC are special cases of CC-Max-$2$-Lin and CC-Max-$2$-Sat respectively, the corresponding hardness results apply to the latter problems as well.

The constants $\beta^{cc}_{vc}$ and $\beta^{cc}_{cut}$ are calculated numerically and their estimated values match the constants $\alpha^{cc}_{2sat}$ and $\alpha^{cc}_{cut}$, which are the approximation ratios for corresponding problems achieved by the algorithm of Raghavendra and Tan \cite{DBLP:conf/soda/RaghavendraT12}. We provide even stronger evidence that these constants match each other, by showing that $\beta^{cc}_{vc}$ and $\beta^{cc}_{cut}$ are calculated as minima of the same functions used for calculating their counterparts $\alpha^{cc}_{vc}$ and $\alpha^{cc}_{cut}$, but over a slightly more restricted domain. 
\par 
Moreover, in Section~\ref{reduction_section} we give refined statements of Theorem \ref{basic_max-cut-sol} and Theorem \ref{basic_max-k-vc} which describe inapproximability of these problems as a function of the cardinality constraint $q \in (0,1)$, which specifies the fraction of variables that need to be set to true. For now, we provide a visualization of these results in Figure~\ref{plot_cc_max_cut}.
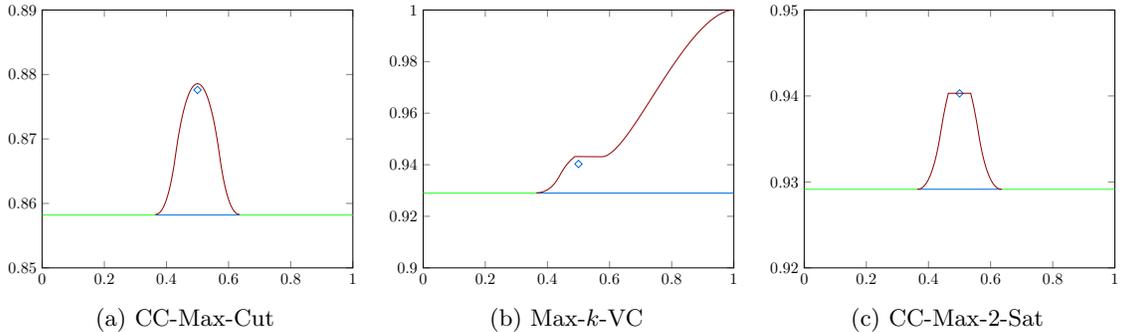
\begin{figure}
  \begin{subfigure}{0.33\textwidth}
    \resizebox{\linewidth}{!}{
          \begin{tikzpicture}
      \begin{axis} [
       xmin=0,xmax=1,
       ymin=0.85,ymax=0.89,
       ]

       \draw[color={rgb:blue,10}] 
          (axis cs:0.365,0.8582 ) -- (axis cs:0.635,0.8582);

       \draw[color={rgb:green,10}] 
          (axis cs:0.0,0.8582 ) -- (axis cs:0.365,0.8582);
       \draw[color={rgb:green,10}] 
          (axis cs:0.635,0.8582 ) -- (axis cs:1.0,0.8582);
        \node[diamond,scale=0.35,draw=blue] at (axis cs:0.5,0.8776) {};

       \draw[color={rgb:red,10}] 
          (axis cs:0.364,0.858296795399) -- 
          (axis cs:0.368,0.858306261673) -- 
          (axis cs:0.372,0.858370262688) -- 
          (axis cs:0.376,0.858491405152) -- 
          (axis cs:0.38,0.858672539132) -- 
          (axis cs:0.384,0.858916787655) -- 
          (axis cs:0.388,0.859227581103) -- 
          (axis cs:0.392,0.859608697354) -- 
          (axis cs:0.396,0.860064308877) -- 
          (axis cs:0.4,0.860599038298) -- 
          (axis cs:0.404,0.861218024352) -- 
          (axis cs:0.408,0.861927000668) -- 
          (axis cs:0.412,0.862732390552) -- 
          (axis cs:0.416,0.863641421881) -- 
          (axis cs:0.42,0.864662267525) -- 
          (axis cs:0.424,0.865804218546) -- 
          (axis cs:0.428,0.867077899934) -- 
          (axis cs:0.432,0.868435564544) -- 
          (axis cs:0.436,0.869694812816) -- 
          (axis cs:0.44,0.87084961795) -- 
          (axis cs:0.444,0.87190708042) -- 
          (axis cs:0.448,0.872873023617) -- 
          (axis cs:0.452,0.873752345392) -- 
          (axis cs:0.456,0.874549130586) -- 
          (axis cs:0.46,0.875266876331) -- 
          (axis cs:0.464,0.875908549825) -- 
          (axis cs:0.468,0.876476657689) -- 
          (axis cs:0.472,0.876973345851) -- 
          (axis cs:0.476,0.877400425848) -- 
          (axis cs:0.48,0.877759379562) -- 
          (axis cs:0.484,0.878051466445) -- 
          (axis cs:0.488,0.878277650918) -- 
          (axis cs:0.492,0.878438694165) -- 
          (axis cs:0.496,0.878535102257) -- 
          (axis cs:0.5,0.878567206395) -- 
          (axis cs:0.504,0.878535102257) -- 
          (axis cs:0.508,0.878438694165) -- 
          (axis cs:0.512,0.878277650918) -- 
          (axis cs:0.516,0.878051466445) -- 
          (axis cs:0.52,0.877759379562) -- 
          (axis cs:0.524,0.877400425848) -- 
          (axis cs:0.528,0.876973345851) -- 
          (axis cs:0.532,0.876476657689) -- 
          (axis cs:0.536,0.875908549825) -- 
          (axis cs:0.54,0.875266876331) -- 
          (axis cs:0.544,0.874549130586) -- 
          (axis cs:0.548,0.873752345392) -- 
          (axis cs:0.552,0.872873023617) -- 
          (axis cs:0.556,0.87190708042) -- 
          (axis cs:0.56,0.87084961795) -- 
          (axis cs:0.564,0.869694812816) -- 
          (axis cs:0.568,0.868435564544) -- 
          (axis cs:0.572,0.867077899934) -- 
          (axis cs:0.576,0.865804218546) -- 
          (axis cs:0.58,0.864662267525) -- 
          (axis cs:0.584,0.863641421881) -- 
          (axis cs:0.588,0.862732390552) -- 
          (axis cs:0.592,0.861927000668) -- 
          (axis cs:0.596,0.861218024352) -- 
          (axis cs:0.6,0.860599038298) -- 
          (axis cs:0.604,0.860064308877) -- 
          (axis cs:0.608,0.859608697354) -- 
          (axis cs:0.612,0.859227581103) -- 
          (axis cs:0.616,0.858916787655) -- 
          (axis cs:0.62,0.858672539132) -- 
          (axis cs:0.624,0.858491405152) -- 
          (axis cs:0.628,0.858370262688) -- 
          (axis cs:0.632,0.858306261673) -- 
          (axis cs:0.636,0.858296795399) ;
      \end{axis}
    \end{tikzpicture}
    }
    \caption{CC-Max-Cut}
    \label{fig intro cc-max-cut}
  \end{subfigure}
  \begin{subfigure}{0.33\textwidth}
    \resizebox{\linewidth}{!}{
        \begin{tikzpicture}
    \begin{axis} [
     xmin=0,xmax=1,
     ymin=0.9,ymax=1.0,
     ]
     \draw[color={rgb:blue,10}] 
      (axis cs:0.365,0.929) -- (axis cs:1.0,0.929);
     \draw[color={rgb:green,10}] 
     (axis cs:0.0,0.929) -- (axis cs:0.365,0.929);
     
     \node[diamond,scale=0.35,draw=blue] at (axis cs:0.5,0.9403) {};

     \draw[color={rgb:red,10}] 
        (axis cs:0.364,0.929148492172) -- 
        (axis cs:0.368,0.929153221764) -- 
        (axis cs:0.372,0.929185218647) -- 
        (axis cs:0.376,0.929245786167) -- 
        (axis cs:0.38,0.929336349347) -- 
        (axis cs:0.384,0.929458469693) -- 
        (axis cs:0.388,0.929613862383) -- 
        (axis cs:0.392,0.929804416345) -- 
        (axis cs:0.396,0.930032217798) -- 
        (axis cs:0.4,0.930299578039) -- 
        (axis cs:0.404,0.930609066418) -- 
        (axis cs:0.408,0.930963549726) -- 
        (axis cs:0.412,0.931366239593) -- 
        (axis cs:0.416,0.931820749927) -- 
        (axis cs:0.42,0.93233116713) -- 
        (axis cs:0.424,0.932902136693) -- 
        (axis cs:0.428,0.933538971061) -- 
        (axis cs:0.432,0.93424778547) -- 
        (axis cs:0.436,0.935035671087) -- 
        (axis cs:0.44,0.935910906052) -- 
        (axis cs:0.444,0.936783924708) -- 
        (axis cs:0.448,0.937595309592) -- 
        (axis cs:0.452,0.938350303593) -- 
        (axis cs:0.456,0.939052938621) -- 
        (axis cs:0.46,0.939704817652) -- 
        (axis cs:0.464,0.940310425404) -- 
        (axis cs:0.468,0.940871563383) -- 
        (axis cs:0.472,0.941389257127) -- 
        (axis cs:0.476,0.941866508244) -- 
        (axis cs:0.48,0.942304612434) -- 
        (axis cs:0.484,0.942704723633) -- 
        (axis cs:0.488,0.943167858755) -- 
        (axis cs:0.576,0.943101014406) -- 
        (axis cs:0.58,0.943158183027) -- 
        (axis cs:0.584,0.943274866722) -- 
        (axis cs:0.588,0.943445784351) -- 
        (axis cs:0.592,0.943666259645) -- 
        (axis cs:0.596,0.943932128873) -- 
        (axis cs:0.6,0.944239665777) -- 
        (axis cs:0.604,0.944585519977) -- 
        (axis cs:0.608,0.944966666028) -- 
        (axis cs:0.612,0.945380360953) -- 
        (axis cs:0.616,0.945824108595) -- 
        (axis cs:0.62,0.9462956295) -- 
        (axis cs:0.624,0.946792835314) -- 
        (axis cs:0.628,0.947313806908) -- 
        (axis cs:0.632,0.947856775569) -- 
        (axis cs:0.636,0.948420106763) -- 
        (axis cs:0.64,0.949002286057) -- 
        (axis cs:0.644,0.94960190684) -- 
        (axis cs:0.648,0.950217659593) -- 
        (axis cs:0.652,0.950848322461) -- 
        (axis cs:0.656,0.951492752938) -- 
        (axis cs:0.66,0.952149880527) -- 
        (axis cs:0.664,0.95281870021) -- 
        (axis cs:0.668,0.953498266648) -- 
        (axis cs:0.672,0.954187688996) -- 
        (axis cs:0.676,0.954886126261) -- 
        (axis cs:0.68,0.955592783139) -- 
        (axis cs:0.684,0.956306906262) -- 
        (axis cs:0.688,0.957027780813) -- 
        (axis cs:0.692,0.957754727464) -- 
        (axis cs:0.696,0.958487099599) -- 
        (axis cs:0.7,0.959224280791) -- 
        (axis cs:0.704,0.959965682503) -- 
        (axis cs:0.708,0.960710741988) -- 
        (axis cs:0.712,0.961458920369) -- 
        (axis cs:0.716,0.962209700879) -- 
        (axis cs:0.72,0.962962587245) -- 
        (axis cs:0.724,0.963717102199) -- 
        (axis cs:0.728,0.964472786106) -- 
        (axis cs:0.732,0.965229195692) -- 
        (axis cs:0.736,0.965985902873) -- 
        (axis cs:0.74,0.966742493663) -- 
        (axis cs:0.744,0.967498567164) -- 
        (axis cs:0.748,0.968253734614) -- 
        (axis cs:0.752,0.969007618513) -- 
        (axis cs:0.756,0.969759851798) -- 
        (axis cs:0.76,0.970510077063) -- 
        (axis cs:0.764,0.971257945841) -- 
        (axis cs:0.768,0.972003117914) -- 
        (axis cs:0.772,0.972745260673) -- 
        (axis cs:0.776,0.973484048499) -- 
        (axis cs:0.78,0.97421916219) -- 
        (axis cs:0.784,0.974950288409) -- 
        (axis cs:0.788,0.975677119155) -- 
        (axis cs:0.792,0.97639935126) -- 
        (axis cs:0.796,0.977116685912) -- 
        (axis cs:0.8,0.977828828184) -- 
        (axis cs:0.804,0.978535486589) -- 
        (axis cs:0.808,0.979236372646) -- 
        (axis cs:0.812,0.979931200457) -- 
        (axis cs:0.816,0.980619686293) -- 
        (axis cs:0.82,0.981301548195) -- 
        (axis cs:0.824,0.981976505572) -- 
        (axis cs:0.828,0.982644278814) -- 
        (axis cs:0.832,0.9833045889) -- 
        (axis cs:0.836,0.983957157013) -- 
        (axis cs:0.84,0.984601704149) -- 
        (axis cs:0.844,0.985237950737) -- 
        (axis cs:0.848,0.985865616237) -- 
        (axis cs:0.852,0.986484418746) -- 
        (axis cs:0.856,0.987094074595) -- 
        (axis cs:0.86,0.98769429793) -- 
        (axis cs:0.864,0.988284800289) -- 
        (axis cs:0.868,0.988865290156) -- 
        (axis cs:0.872,0.989435472511) -- 
        (axis cs:0.876,0.989995048344) -- 
        (axis cs:0.88,0.990543714163) -- 
        (axis cs:0.884,0.991081161461) -- 
        (axis cs:0.888,0.991607076164) -- 
        (axis cs:0.892,0.992121138034) -- 
        (axis cs:0.896,0.992623020045) -- 
        (axis cs:0.9,0.993112387694) -- 
        (axis cs:0.904,0.993588898282) -- 
        (axis cs:0.908,0.994052200114) -- 
        (axis cs:0.912,0.994501931642) -- 
        (axis cs:0.916,0.994937720521) -- 
        (axis cs:0.92,0.995359182574) -- 
        (axis cs:0.924,0.995765920646) -- 
        (axis cs:0.928,0.99615752333) -- 
        (axis cs:0.932,0.99653356354) -- 
        (axis cs:0.936,0.99689359691) -- 
        (axis cs:0.94,0.997237159968) -- 
        (axis cs:0.944,0.997563768062) -- 
        (axis cs:0.948,0.997872912961) -- 
        (axis cs:0.952,0.99816406007) -- 
        (axis cs:0.956,0.998436645155) -- 
        (axis cs:0.96,0.998690070452) -- 
        (axis cs:0.964,0.998923699973) -- 
        (axis cs:0.968,0.999136853754) -- 
        (axis cs:0.972,0.999328800681) -- 
        (axis cs:0.976,0.999498749306) -- 
        (axis cs:0.98,0.999645835758) -- 
        (axis cs:0.984,0.999769107204) -- 
        (axis cs:0.988,0.999867497994) -- 
        (axis cs:0.992,0.999939792516) -- 
        (axis cs:0.996,0.999984559607) --
        (axis cs:1.0,1.0) ;
    \end{axis}
  \end{tikzpicture}
    }
    \caption{Max-$k$-VC}
    \label{fig intro max-k-vc}
  \end{subfigure}
  \begin{subfigure}{0.33\textwidth}
    \resizebox{\linewidth}{!}{
      		\begin{tikzpicture}

        \begin{axis} [
         xmin=0,xmax=1,
         ymin=0.92,ymax=0.95,
         ]

         \draw[color={rgb:blue,10}] 
          (axis cs:0.365,0.929148492172) -- (axis cs:0.635,0.929148492172);
         \draw[color={rgb:green,10}] 
         (axis cs:0.0,0.929148492172) -- (axis cs:0.365,0.929148492172);
         \draw[color={rgb:green,10}] 
         (axis cs:0.635,0.929148492172) -- (axis cs:1.0,0.929148492172);
        
         \node[diamond,scale=0.35,draw=blue] at (axis cs:0.5,0.9403) {};
         \draw[color={rgb:red,10}] 
            (axis cs:0.364,0.929148492172) -- 
            (axis cs:0.368,0.929153221764) -- 
            (axis cs:0.372,0.929185218647) -- 
            (axis cs:0.376,0.929245786167) -- 
            (axis cs:0.38,0.929336349347) -- 
            (axis cs:0.384,0.929458469693) -- 
            (axis cs:0.388,0.929613862383) -- 
            (axis cs:0.392,0.929804416345) -- 
            (axis cs:0.396,0.930032217798) -- 
            (axis cs:0.4,0.930299578039) -- 
            (axis cs:0.404,0.930609066418) -- 
            (axis cs:0.408,0.930963549726) -- 
            (axis cs:0.412,0.931366239593) -- 
            (axis cs:0.416,0.931820749927) -- 
            (axis cs:0.42,0.93233116713) -- 
            (axis cs:0.424,0.932902136693) -- 
            (axis cs:0.428,0.933538971061) -- 
            (axis cs:0.432,0.93424778547) -- 
            (axis cs:0.436,0.935035671087) -- 
            (axis cs:0.44,0.935910906052) -- 
            (axis cs:0.444,0.936783924708) -- 
            (axis cs:0.448,0.937595309592) -- 
            (axis cs:0.452,0.938350303593) -- 
            (axis cs:0.456,0.939052938621) -- 
            (axis cs:0.46,0.939704817652) -- 
            (axis cs:0.464,0.940310425404) -- 
            (axis cs:0.536,0.940310425404) -- 
            (axis cs:0.54,0.939704817652) -- 
            (axis cs:0.544,0.939052938621) -- 
            (axis cs:0.548,0.938350303593) -- 
            (axis cs:0.552,0.937595309592) -- 
            (axis cs:0.556,0.936783924708) -- 
            (axis cs:0.56,0.935910906052) -- 
            (axis cs:0.564,0.935035671087) -- 
            (axis cs:0.568,0.93424778547) -- 
            (axis cs:0.572,0.933538971061) -- 
            (axis cs:0.576,0.932902136693) -- 
            (axis cs:0.58,0.93233116713) -- 
            (axis cs:0.584,0.931820749927) -- 
            (axis cs:0.588,0.931366239593) -- 
            (axis cs:0.592,0.930963549726) -- 
            (axis cs:0.596,0.930609066418) -- 
            (axis cs:0.6,0.930299578039) -- 
            (axis cs:0.604,0.930032217798) -- 
            (axis cs:0.608,0.929804416345) -- 
            (axis cs:0.612,0.929613862383) -- 
            (axis cs:0.616,0.929458469693) -- 
            (axis cs:0.62,0.929336349347) -- 
            (axis cs:0.624,0.929245786167) -- 
            (axis cs:0.628,0.929185218647) -- 
            (axis cs:0.632,0.929153221764) -- 
            (axis cs:0.636,0.929148492172) ;
      \end{axis}
    \end{tikzpicture}
    }
    \caption{CC-Max-$2$-Sat}
    \label{fig intro cc-max-2-sat}
  \end{subfigure}
  \caption{Hardness ratio (${\color{red}-}$), approximation ratio ({\color{blue}$-$},{\color{blue}$\diamond$}), and matching \mbox{algorithm}/hardness ratio (${\color{green}-}$) as a function of the cardinality constraint $q$ for the three problems.}
  \label{plot_cc_max_cut}
\end{figure}
\par
\paragraph{Overview of proof ideas}

The main observation behind the hardness results is that the reduction used to prove hardness of approximation for the Independent Set and Vertex Cover problems in bounded degree graphs \cite{DBLP:journals/toc/AustrinKS11} gives very strong soundness guarantees.  In particular it shows that in the ``no'' case of the reduction, all induced subgraphs of the graph contain many edges, which in turn gives useful upper bounds on the number of edges cut by a bipartition of a given size, or the number of edges covered by a subgraph. This is also how \cite{DBLP:conf/soda/Manurangsi19} obtained the previous hardness of $\approx 0.944$ for Max-$k$-VC.
Thus our results use essentially the same reduction as \cite{DBLP:journals/toc/AustrinKS11} (which is in turn similar to the reduction for Max-Cut \cite{DBLP:journals/siamcomp/KhotKMO07}).  Note however that even though the graph produced by that reduction has a small vertex cover in the ``yes'' case, using that small vertex cover is not necessarily the best solution for the Max-$k$-VC problem on the graph.  In particular for $q < 1/2$, it makes more sense to instead use the large independent set as the Max-$k$-VC solution in the yes case (the intuition being that since it is independent, it covers many edges relative to its size).
\par 
Another difference is that we have somewhat greater flexibility in choosing the noise distribution of our ``dictatorship test'' (the key component of essentially all UG-hardness results)
The reason is that for Independent Set/Vertex Cover, the reduction needs ``perfect completeness'', i.e., in the ``yes'' case it needs to produce graphs with large independent sets/small vertex covers, whereas for e.g.~Max-$k$-VC we are perfectly happy with graphs where there are sets of size $k$ covering many, but not necessarily all, edges.
This increased flexibility turns out to improve the hardness ratios for some range of the cardinality constraint $q$.
For example, for the CC-Max-Cut problem with $q=1/2$, this allows us to recover the $\alpha_{GW}$-hardness for the Max-Bisection problem using the same reduction. However, at $q$ further away from $1/2$, and in particular at the minima in Figure~\ref{fig intro cc-max-cut}, it turns out that this flexibility does not help.
E.g.~in the global minimum at $q \approx 0.365$ for Max-$k$-VC, the reduction does output a graph with a large independent set containing a $q$ fraction of the vertices, and choosing that independent set is the optimal solution for the Max-$k$-VC instance.  Similarly, in the part of the range $q > 1/2$ where the function is increasing, the optimal solution to the Max-$k$-VC instance in the yes case is to pick an actual vertex cover of size $q$, and the first point of the curve in this range corresponds exactly to the hardness of $0.944$ from \cite{DBLP:conf/soda/Manurangsi19}.

\paragraph{Organization}
This paper is organized as follows. In Section \ref{preliminaries_section} we fix the notation, recall some well-known facts, and formally introduce the problems of interest. In Section \ref{reduction_section} we give our main hardness reduction and improved inapproximability results.  In Section \ref{RT_section} we give a brief overview of the algorithm of Raghavendra and Tan \cite{DBLP:conf/soda/RaghavendraT12} in order to observe that the hardness ratios we get match the approximation ratios of the algorithm. Finally, in Section \ref{discussion_section} we propose some possible directions for future research.

\section{Preliminaries} \label{preliminaries_section}
\subsection{Notational Conventions}

In this paper we work with undirected (multi)graphs $G = (V, E)$. For a set $S \subseteq V$ of vertices we use $S^c$ to denote its complement $S^c = V \setminus S$, and write $U \sqcup V$ for a disjoint union of sets $U$ and $V$. 
The graphs are both edge and vertex weighted and the weights of vertices and edges are given by functions $w\colon V \rightarrow [0,1]$, and $w\colon E \rightarrow [0,1]$. For subsets $S \subseteq V$ and $K \subseteq E$ we interpret $w(S)$ and $w(K)$ as the sum of weights of vertices contained in $S$ and edges in $K$, respectively. Furthermore, weights are normalized so that $w(V) = w(E) =1 $ and the weight of each vertex equals half the weight of all edges adjacent to it. Therefore, the weights of edges and vertices can be interpreted as probability distributions, and sampling a vertex with probability equal to its weight is the same as sampling an edge and then sampling one of its endpoints with probability $1/2$. For $S,T \subseteq V$, we write $w(S,T)$ for the total weight of edges from $E$ which have one endpoint in $S$, and other in $T$. Note that, since we work with undirected graphs, the order of endpoints is not important, and therefore $w(S,T) = w(T,S)$.
In other words, the weight of an edge $e=(u,v)$ contributes to $w(S,T)$ if either $(u,v) \in T\times S$ or $(u,v) \in S \times T$. We also have the identity
\begin{equation}
  \label{eqn:weight relation}
  w(S, V) = w(S) + \frac{1}{2} w(S, S^c).
\end{equation}
The set of all neighbours of a vertex $v$ including $v$ is denoted by $N(v)$, and the set of all neighbours of a set $S \subseteq V$ including $S$ is denoted by $N(S)$. Let us also introduce the following definition.
\begin{defn}
  A graph $G$ is $(q,\varepsilon)$-dense if every subset $S \subseteq V$ with $w(S) = q$ satisfies $w(S,S) \geq \varepsilon$.
\end{defn}

We use $\phi(x) = \frac{1}{\sqrt{2\pi}}e^{-x^2/2}  $ to denote the density function of a standard normal random variable, and $\Phi(x) = \int_{-\infty} ^{x} \phi(y)dy$ to denote its cumulative distribution function (CDF). We also work with bivariate normal random variables, and to that end introduce the following function.

\begin{defn}
  Let $\rho \in [-1,1] $, and consider two jointly normal random variables $X,Y$ with mean $0$ and covariance matrix $\Cov(X,Y) = \begin{bmatrix}1  & \rho \\ \rho  & 1   \end{bmatrix}$. We define $\Gamma_{\rho} \colon [0,1]^2 \rightarrow [0,1] $ as 
  \begin{equation*}
    \Gamma_{\rho}(x,y) = \Pr \left [ X \leq \Phi^{-1}(x) \wedge Y \leq \Phi^{-1}(y) \right ]. 
  \end{equation*}
\end{defn}
We also write $\Gamma_{\rho}(x) = \Gamma_{\rho}(x,x)$. We have the following basic lemma (for a proof see Appendix A of \cite{DBLP:conf/soda/AustrinBG13}).
\begin{lemma} \label{bivariate_lemma}
  For every $\rho \in [-1,1] $, and every $x,y \in [0,1] $, we have 
  \begin{equation*}
    \Gamma_{\rho}(x,y)= \Gamma_{\rho}(1-x,1-y) -1+x+y.
  \end{equation*}
\end{lemma}

\subsection{Problem Definitions}
This paper is concerned with Max-Cut, Max-$2$-Lin, Max-$2$-Sat, and Max-$k$-VC problems with cardinality constraints. Let us give the definitions of these problems as integer optimization programs now. In these definitions instead of $\{0,1\}$ we represent Boolean domain as $\{-1,1\}$, and for that reason instead of cardinality constraint $q$ we consider a \emph{balance constraint} $r= 1- 2q$.

\begin{defn}\label{Max-Cut_definition}
  An instance $\mathcal{F}$ of the cardinality constrained Max-$2$-Lin (\emph{CC-Max-$2$-Lin}) problem with balance constraint $r \in (-1,1)$ over variables $X=\{x_1,\hdots,x_{n}\}$ taking values in $\{-1,1\}$  is given by the following integer optimization program
  \begin{equation*}\label{max_e2lin_integer_equation}
    \begin{split}
      \max \quad &  \sum\limits_{(i,j)=e_{\ell} \in E } \frac{1+P_{\ell} x_i x_j}{2}, \\
      \textrm{s.t.} \quad & \sum\limits_{i \in V } x_i = nr,
    \end{split}
  \end{equation*}
  where $P_{\ell} \in \{-1,1\}$ and the term $(1+P_{\ell}x_ix_j)/2$ corresponds to the XOR constraint $x_i  x_j =P_{\ell}$. In case $P_{\ell}=-1 $ for all $\ell$, the integer optimization program is an instance of \emph{CC-Max-Cut} problem.
\end{defn}

\begin{defn}\label{max-2-sat-cc_definition}
  An instance $\mathcal{F}$ of the cardinality constrained Max-$2$-Sat (\emph{CC-Max-$2$-Sat}) problem with balance constraint $r \in (-1,1)$ over variables $X=\{x_1,\hdots,x_{n}\}$ taking values in $\{-1,1\}$ is given by the following integer optimization program
  \begin{equation*} \label{max2_sat_integer_equation}
    \begin{split}
      \max \quad & \sum\limits_{(i,j)=e_{\ell} \in E } \frac{3+P_{\ell}^1 x_i+P_{\ell}^2x_j+P_{\ell}^{3}x_ix_j}{4}, \\
      \textrm{s.t.} \quad & \sum\limits_{i \in V } x_i = nr,
    \end{split}
  \end{equation*}
  where $(P_{\ell}^1, P_{\ell}^2, P_{\ell}^3) \in \{ (-1,-1,-1)$, $(1,-1,1),$ $(-1,1,1),$ $(1,1,-1) \}$ corresponds to one of the four possible clauses
  \begin{equation*}
    x_i \vee x_j, \quad \neg x_i \vee x_j, \quad x_i \vee \neg x_j, \quad \neg x_i \vee \neg x_j.
  \end{equation*}
  In case $(P_{\ell}^1, P_{\ell}^2, P_{\ell}^3)=(-1,-1,-1)$ for all $\ell$, the integer optimization program is an instance of \emph{Max-$k$-VC} problem.
\end{defn}

The objective in the problems given by Definitions \ref{Max-Cut_definition} and \ref{max-2-sat-cc_definition} is to find an assignment $z\colon X \rightarrow \{-1,1\}$ which satisfies a (hard) global cardinality constraint and maximizes the number of satisfied soft constraints represented by the objective function. For an assignment $z$ that satisfies global constraint of an instance $\mathcal{F}$ we use $\val_z(\mathcal{F})$ to denote the value of the objective function under the assignment $z$. Furthermore, we use
\[\optval(\mathcal{F}) = \max_{\substack{z\colon X \rightarrow \{-1,1\}\\\sum_{x \in X} z(x) = rn}} \val_z(\mathcal{F})\]
to denote the maximum value of the objective function over all assignments $z$ satisfying the cardinality constraint. 
\par 
The starting point of the hardness results in this paper is the Unique Games problems, which is defined as follows.
\begin{defn}
  A \emph{Unique Games instance} $\Lambda = (\mathcal{U},\mathcal{V},\mathcal{E},\Pi,[L])$ consists of an unweighted bipartite multigraph 
  $(\mathcal{U} \sqcup \mathcal{V},\mathcal{E})$, a set 
  $\Pi = \{\pi_{e}\colon [L] \to [L] \mid e \in
  \mathcal{E} \textrm{ and }\pi_e \textrm{ is a bijection} \} $ of permutation constraints, and a set [L] of labels. The value of $\Lambda$ under the assignment $z\colon \mathcal{U} \sqcup \mathcal{V} \to [L] $ is the fraction of edges satisfied, where an edge $e=(u,v), u \in \mathcal{U}, v \in \mathcal{V}$ is satisfied if $\pi_e(z(u)) = z(v)$. We write $\val_c(\Lambda)$ for the value of $\Lambda$ under $z$, and $\opt(\Lambda)$ for the maximum possible value over all assignments $z$.
\end{defn}

The Unique Games Conjecture \cite{DBLP:conf/stoc/Khot02a} can be formulated as follows (\cite{DBLP:conf/coco/KhotR03}, Lemma 3.4).
\begin{conj}[Unique Games Conjecture] 
  For every constant $\gamma > 0$ there is a sufficiently large $L \in \mathbb{N}$, such that for a Unique Games instance $\Lambda = (\mathcal{U},\mathcal{V},\mathcal{E},\Pi,[L]) $ with a regular bipartite graph $(\mathcal{U} \sqcup \mathcal{V}, \mathcal{E})$, it is NP-hard to distinguish between
  \begin{itemize}
    \item $\opt(\Lambda) \geq 1-\gamma$, 
    \item $\opt(\Lambda) \leq \gamma$.
  \end{itemize}
\end{conj}

\subsection{Analysis of Boolean Functions}
One of the ubiquitous tools in the hardness of approximation area is Fourier analysis of Boolean functions. We now recall some of the well-known facts which are used in the paper. For a more detailed study, we refer to \cite{DBLP:books/daglib/0033652}.
\par 
For $q\in [0,1]$ and $n \in \mathbb{N}$ we write $\pi_q\colon \{0,1\} \to [0,1] $ for the probability distribution given by $\pi_q(1)=q, \pi_q(0)=1-q$. We also write $\pi_q^{\otimes n}$ for the probability distribution on $n$-bit strings $x \in \{0,1\}^n$ where each bit is distributed according to $\pi_q$, independently. We use $L^2(\pi_q^{\otimes n})$ to denote the space of random variables $f\colon \{0,1\}^n  \to \mathbb{R}$ over the probability space $\left( \{0,1\}^n, P(\{0,1\}^n), \pi_q^{\otimes n } \right)$, and interpret $\mathop{\mathbf{E}}[f] $ and $\mathop{\mathbf{Var}}[f] $ as expectation and variance of $f(X)$ when the $X$ is drawn from $\pi_q^{\otimes n}$. Depending on context, the elements of $L^2(\pi_q^{\otimes n})$ will be interpreted as functions as well.
\par 
Let us now introduce some of the common objects used in the study of Boolean functions.
\begin{defn}
  Consider a function $f \in L^2(\pi_q^{\otimes n})$ and $i \in \{1,\hdots,n\}$. The \emph{influence} $\mathbf{Inf}_i[f]$ of the $i$-th argument on $f$ is defined as 
  \begin{equation*}
    \mathbf{Inf}_i[f] = \mathbf{E}_{x \sim \pi_q^{\otimes n} } [ \mathbf{Var} _{\tilde{x}_i \sim \pi_q} [ f(x_1,\hdots,x_{i-1},\tilde{x}_i,x_{i+1},\hdots,x_n)] ]. 
  \end{equation*}
\end{defn}

Minimal correlation between two $q$-biased bits is $\max(-q/(1-q),-(1-q)/q)$. For notational convenience, let us introduce the function $\kappa$  which assigns to each value $q \in (0,1)$ an interval $I \subseteq (-1,0)$ as
\begin{equation*}
  \kappa(q) = \begin{cases}
               [-q/(1-q),0), & \text{if } q<1/2, \\
               (-1,0),       & \text{if } q = 1/2, \\
               [-(1-q)/q,0), & \text{if } q >1/2.
              \end{cases}
\end{equation*}

\begin{defn}
  For a fixed $x\in \{0,1\}, q \in (0,1)$ and $\rho \in \kappa(q)$ we write $y \sim N_{\rho}(x)$ to indicate that $y$ is a $\rho$-correlated copy of $x$.  In particular each bit $y_i$ is equal to $1$ with probability $q+\rho(1-q)$ if $x_i=1$, and $y_i=1$ with probability $q-\rho q $ when $x_i = 0$, independently.
\end{defn}

\begin{defn}
  Consider $q \in (0,1)$ and $\rho \in \kappa(q)$. The \emph{noise operator} $T_{\rho} \colon L^2(\pi_q^{\otimes n}) \to L^2( \pi_q^{\otimes n})$ is defined as 
  \begin{equation*}
    T_{\rho} f (x) = \mathbf{E}_{y \sim N_{\rho}(x)}[f(y)].
  \end{equation*}
\end{defn}

The following lemma gives a useful bound on the number of influential variables of $T_{\rho}f$. 
\begin{lemma} \label{inf_bound_lemma}
  Consider $q \in (0,1)$, a function $f \in L^2(\pi_q^{\otimes n})$, and $\rho \in \kappa(q)$. Then, for any $\tau > 0 $ we have that
  \begin{equation*}
    | \{ i \in [n] \mid \mathbf{Inf}_i[T_{\rho} f] \geq \tau \} | \leq \frac{\mathbf{Var}[f] }{\tau e \ln(1/|\rho|)}.
  \end{equation*}
\end{lemma}
For a proof we refer to Lemma 3.4 of \cite{DBLP:conf/focs/GuruswamiMR08}. We also need to introduce the notion of noise stability, defined as follows.

\begin{defn}
  Let $q\in (0,1), \rho \in \kappa(q)$ and $f \in L^2(\pi_q^{\otimes n})$. The \emph{noise stability} of function $f$ at $\rho$ is defined as
  \begin{equation*}
    \mathbb{S}_{\rho} = \mathbf{E}[ f \cdot T_{\rho}f] .
  \end{equation*}
\end{defn}

Let us also recall the following variant of the ``Majority is Stablest'' theorem in the form that appeared in \cite{DBLP:journals/toc/AustrinKS11}, and which follows from Theorem 3.1 in \cite{DBLP:journals/siamcomp/DinurMR09}.
\begin{thm} \label{thresholds_are_stablest}
  Let $q\in (0,1)$ and $\rho \in \kappa(q)$. Then for any $\varepsilon > 0 $, there exist $\tau>0$ and $\delta > 0 $ such that for every function $f \in L^2(\pi_q^{\otimes n})$, $f\colon \{-1,1\}^n \rightarrow [0,1]$ that satisfies
  \begin{equation*}
    \max_{i\in [n] } \mathbf{Inf}_i[T_{1-\delta}f] \leq \tau,
  \end{equation*}
  we have 
  \begin{equation*}
    \mathbb{S}_{\rho}(f) \geq \Gamma_{\rho}(\mathbf{E}[f] ) - \varepsilon.
  \end{equation*}
\end{thm}

\section{Hardness Reduction} \label{reduction_section}
In this section we give our main hardness reduction.  As discussed in the introduction, it is a generalization of the reduction of Theorem III.1 from \cite{DBLP:journals/toc/AustrinKS11}.
\begin{thm}\label{hardness_cc_max_cut}
  For every $q \in (0,1), \varepsilon > 0$, and $\rho \in \kappa(q)$, there exists a $\gamma > 0 $ and a reduction from Unique Games instances $\Lambda=(\mathcal{U},\mathcal{V},\mathcal{E},\Pi,[L])$ to weighted multigraphs $G = (V, E)$ with the following properties:
  \begin{itemize}
  \item \emph{Completeness:} If $\opt(\Lambda) \geq 1-\gamma$, then there is a set $S \subseteq V$ such that $w(S) = q$ and $w(S,S^c) \geq 2q(1-q)(1-\rho) - 2\gamma$.
    \item \emph{Soundness:} If $\opt(\Lambda) \leq \gamma$, then for every $r \in [0,1]$, $G$ is $(r, \Gamma_{\rho}(r)-\varepsilon)$-dense.
  \end{itemize}
  Moreover, the running time of the reduction is polynomial in $|\mathcal{U}|,|\mathcal{V}|,|\mathcal{E}|$, and exponential in $L$.
  \begin{proof}
    Let $\nu\colon \{0,1\}^2 \rightarrow [0,1]$ be the probability distribution over two $\rho$-correlated $q$-biased bits.  In other words, letting $t= (q-q^2)(1-\rho)$, we have
    \begin{equation*} 
      \nu(0,0)=1-q-t, \quad \nu(0,1) = \nu(1,0) = t, \quad \nu(1,1)= q-t.
    \end{equation*}
    \par 
    Let us now describe how the multigraph $G$ can be constructed from $\Lambda$. We define the vertex set of $G$ to be $V = \mathcal{V} \times \{0,1\}^L = \{(v,x) \mid v \in \mathcal{V}, x \in \{0,1\}^L\}$. In particular, for every vertex $v \in \mathcal{V}$ we create $2^L$ vertices of $G$, which we identify with $L$-bit strings in $\{0,1\}^L$. We also write $v^x$ for a vertex $(v,x)$ of the graph $G$. The weights of vertices in $G$ are given by
    \begin{equation}
      w(v^x) = \frac{1}{|\mathcal{V}|} \pi_q^{\otimes L}(x).
    \end{equation}
    The edges of $G$ are constructed in the following way. For every $u\in \mathcal{U}$, and for every two $v_1,v_2\in N(u)$, we create an edge between vertices $v_1^x, v_2^y$ with weight
    \begin{equation*}
      \frac{1}{|\mathcal{U}|D^2} \nu^{\otimes L}(x \circ \pi_{e_1}, y \circ \pi_{e_2}), \quad  \textrm{where } e_1 = (u,v_1), \quad e_2 = (u,v_2).
    \end{equation*}
    \par 
    Expressed formally, the edge set $E$ is
    \begin{equation*}
      E = \{(e_1^x,e_2^y) \mid e_1 = (u,v_1), e_2 = (u,v_2), u \in \mathcal{U}, v_1,v_2 \in \mathcal{V}, x,y \in \{ 0, 1 \}^L \}.
    \end{equation*}
    Since the marginal of the distribution $\nu$ over either the first or the second argument is a $q$-biased distribution on $\{0,1\}^L$, the weight of all edges adjacent to a vertex $v^x$ equals two times the weight of the vertex $v^x$. Furthermore, it is trivial to check that $w(V) = w(E) = 1 $.
    The number of vertices in $G$ is $|\mathcal{V}|2^L$, and the number of edges is $|\mathcal{U}|D^22^L$, so the construction is indeed polynomial in $|\mathcal{U}|,|\mathcal{V}|$ and $|\mathcal{E}|$. 
    \par
    Let us now prove completeness and soundness.
    \par 
    \emph{Completeness:}
    Since $\opt(\Lambda) \geq 1-\gamma$, there is a labeling $z \colon \mathcal{U} \sqcup \mathcal{V} \to [L] $ such that $\val_z(\Lambda)\geq 1-\gamma$. Consider a set $S$ given by
    \begin{equation*}
      S = \{ v^x \in V \mid x_{z(v)} = 1 \}.
    \end{equation*}
    The weight of the set $S$ is obviously $q$. Let us consider a set consisting of pairs of edges in $\mathcal{E}$ which have a common vertex in $\mathcal{U}$, i.e. the set
    \begin{equation*}
      \hat{E} =\{(e_1,e_2) \in \mathcal{E}\times \mathcal{E} \mid e_1 = (u,v_1), e_2=(u,v_2), u\in \mathcal{U}, v_1,v_2 \in \mathcal{V} \},
    \end{equation*}
    and its subset $\hat{E}_{\text{good}}$ consisting of edge pairs which are satisfied under the assignment $z$, or formally
    \begin{equation*}
      \hat{E}_{\text{good}} = \{(e_1,e_2) \in \hat{E} \mid e_1 = (u,v_1), e_2=(u,v_2), z(u) = \pi^{-1}_{e_1}(z(v_1)) = \pi^{-1}_{e_2}(z(v_2)) \},
    \end{equation*}
    Since at least fraction $1-\gamma$ of edges in $\mathcal{E}$ are satisfied under $z$, at least fraction $(1-\gamma)^2$ of edge pairs in $\hat{E}$ is satisfied under $z$, i.e. $|\hat{E}_{\text{good}}|\geq (1-\gamma)^2 |\hat{E}|$. For every $(e_1,e_2) \in \hat{E}_{\text{good}}, e_1=(u,v_1), e_2 = (u,v_2)$, the edges between $S$ and $S^c$ created through the pair of edges $(e_1,e_2)$ have the total weight of 
    \begin{equation*}
      \begin{split}
        \frac{1}{|\mathcal{U}|D^2} \Pr_{(x,y) \sim \nu^{\otimes L} } \left [ (x \circ \pi_{e_1^{-1}})_{z(v_1)} \neq (y \circ \pi_{e_2^{-1}})_{z(v_2)} \right ]  & 
        = \frac{1}{|\mathcal{U}|D^2} \Pr_{(x,y) \sim \nu^{\otimes L} } \left [ x_{z(u)} \neq y_{z(u)} \right ]  \\ & =      \frac{1}{|\mathcal{U}|D^2} \left(\nu(0,1)+\nu(1,0) \right)=\frac{1}{|\mathcal{U}|D^2}2t.
      \end{split}
    \end{equation*} 
    Therefore, we have $w(S,S^c) \geq 2t(1-\gamma)^2 \ge 2q(1-q)(1-\rho) - 2 \gamma$. 
    \par 
    \emph{Soundness:} 
    Let us assume by contradiction that $G$ is not $\left(r,\Gamma_{\rho}(r)-\varepsilon\right)$-dense, and therefore that there is a set $S \subseteq V$ of weight $w(S) = r$ for which $w(S,S) < \Gamma_\rho(r) - \varepsilon$. For each $v \in \mathcal{V}$, let us define a function $S_v \in L^2(\pi_q^{\otimes L})$ to be the indicator function of $S$ restricted to the vertex $v$. In particular, we have that $S_v(x) = 1 $ if and only if $v^x \in S$. Furthermore, for all $u \in \mathcal{U}$ let us define $S_u \in L^2(\pi_q^{\otimes L})$ as
    \begin{equation*}
      S_u(x) = \mathop{\mathbf{E}}\limits_{\substack{e=(u,v),  \\ v \in N(u) }} [ S_v(x \circ \pi_e^{-1}) ] .
    \end{equation*}
  We have that 
  \begin{equation*}
    \begin{split}
      w(S,S) & = \mathop{\mathbf{E}}\limits_{\substack{u \in \mathcal{U}, \\ e_1=(u,v_1), e_2 = (u,v_2)\\ v_1,v_2 \in N(u) }} \left [ \mathop{\mathbf{E}}_{(x,y) \sim \nu^{\otimes L}}[ S_{v_1}(x \circ \pi_{e_1}^{-1}) S_{v_2}(y \circ \pi_{e_2}^{-1} ) ]  \right ]   \\
             & = \mathop{\mathbf{E}}\limits_{\substack{u \in \mathcal{U}, \\ (x,y) \sim \nu^{\otimes L} }} \left [ \mathop{\mathbf{E}}_{\substack{e_1=(u,v_1), e_2 = (u,v_2)\\ v_1,v_2 \in N(u)}}[ S_{v_1}(x \circ \pi_{e_1}^{-1}) S_{v_2}(y \circ \pi_{e_2}^{-1} ) ]  \right ]    \\ 
             & = \mathop{\mathbf{E}}\limits_{\substack{u \in \mathcal{U} }} \left [ \mathop{\mathbf{E}}_{(x,y) \sim \nu^{\otimes L} }[S_u(x) S_u(y)] \right] = \mathop{\mathbf{E}}\limits_{\substack{u \in \mathcal{U} }} \left [  \mathop{\mathbf{E}}_{x \sim \pi_q^{\otimes L}}[S_u(x) T_{\rho} S_u(x)] \right ]  = \mathop{\mathbf{E}}\limits_{\substack{u \in \mathcal{U} }} [\mathbb{S}_{\rho}(S_u) ] .
    \end{split}
  \end{equation*}
  Let us define $\mu_u = \mathop{\mathbf{E}}_{x \sim \pi_q^{\otimes L} } \left [ S_u(x) \right] $, and remark that due to regularity of $\Lambda$ we have $\mathop{\mathbf{E}}_{u \in U } \left [ S_u \right] = r$.  We claim that there is a set $\mathcal{U}' \subseteq \mathcal{U}, |\mathcal{U}'| \geq \varepsilon |\mathcal{U}|/2$ such that for every $u \in \mathcal{U}'$ we have $\mathbb{S}_{\rho}(S_u) < \Gamma_{\rho}(\mu_u)-\varepsilon/2$. Otherwise, we reach a contradiction by noticing that
    \begin{equation*} 
      \Gamma_{\rho}(r) - \varepsilon >  w(S,S)  =  \mathop{\mathbf{E}}\limits_{\substack{u \in \mathcal{U} }} [\mathbb{S}_{\rho}(S_u) ]  \geq (1-\varepsilon/2 ) \left ( \mathop{\mathbf{E}}_{u \in U} \left [ \Gamma_{\rho}(\mu_u)\right ] -\varepsilon/2 \right) \geq \mathop{\mathbf{E}}_{u \in U} \left [  \Gamma_{\rho}(\mu_u) \right ]  - \varepsilon \geq \Gamma_{\rho}(r) - \varepsilon,
    \end{equation*}
    where in the last inequality we used the fact that $\Gamma_{\rho}$ is convex. 
    \par  
    By Theorem  \ref{thresholds_are_stablest} there is $\tau>0$ and $\delta>0$ such that for every $u \in \mathcal{U}'$ there is a significant coordinate $i \in [L]$ for which $\mathbf{Inf}_i[T_{1-\delta} S_u] \geq \tau$. For each $u \in \mathcal{U}'$ and for its significant coordinate $i$, by using the fact that $\mathbf{Inf}_i$ is convex and Markov's inequality we conclude that for at least $\tau/2$ of $v \in N(u)$ we have 
    \begin{equation*}
      \mathbf{Inf}_{\pi_e(i)} [ T_{1-\delta} S_v ] \geq \tau /2 , \quad e=(u,v).
    \end{equation*}
    For each $v \in \mathcal{V}$ let $[L]_v \subseteq [L] $ denote a set of labels defined by
    \begin{equation*}
      [L]_v = \{ i \in [L] \mid \mathbf{Inf}_{i} [ T_{1-\delta} S_v ] \geq \tau /2 \}.
    \end{equation*}
    By Lemma \ref{inf_bound_lemma} we have that $|[L]_v| \leq \frac{2}{\tau e \ln(1/(1-\delta) )}$. Let us now pick an assignment $z\colon \mathcal{U} \sqcup \mathcal{V} \to [L] $ of $\Lambda$ using the following randomized procedure. For each $v \in \mathcal{V}$, pick $i \in [L]_v$ randomly, and set $z(v)=i$. If $[L]_v = \emptyset$, we pick $i\in [L]$ randomly. Then, for each $u \in \mathcal{U}$, we set $z(u)=i$ for the $i$ that maximizes the number of edges satisfied. From the previous discussion we conclude that this labeling satisfies $\Omega\left(\varepsilon \tau^4 \ln^2(1/(1-\delta)) \right)$ of constraints of $\Lambda$ in expectation. But since this constant does not depend on $\gamma$ this would be a contradiction if we started with a sufficiently small $\gamma$.
  \end{proof}
\end{thm}

\subsection{Hardness for CC-Max-Cut}\label{cc_max_cut_hardness_section}
Now that we have proven Theorem \ref{hardness_cc_max_cut}, it is straightforward to prove the following theorem which gives a hardness result of CC-Max-Cut. 
\begin{thm}\label{max-cut-sol}
  For any $q \in (0,1)$ and $\rho \in \kappa(q)$ it is UG-hard to approximate CC-Max-Cut with cardinality constraint $q$ within $\beta^{cc}_{cut}(q,\rho) + \varepsilon$ where $\varepsilon>0$ is arbitrary small and $\beta^{cc}_{cut}(q,\rho)$ is given by
  \begin{equation*}
    \beta^{cc}_{cut}(q,\rho) = \frac{1-\Gamma_{\rho}(q)-\Gamma_{\rho}(1-q) }{2(q-q^2)(1-\rho)}.
  \end{equation*}
  \begin{proof}
    By Theorem~\ref{hardness_cc_max_cut} there exists a family of multigraphs $G=(V,E)$ for which it is UG-hard to decide between the following two statements:
    \begin{itemize}
      \item There is a set $S\subseteq V, w(S) = q$, such that $w(S,S^c) \geq 2q(1-q) (1-\rho) -2\gamma$.
      \item For any $r \in [0,1]$ and every set $T\subseteq V, w(T) = r$ we have $w(T,T) \geq \Gamma_{\rho} (r) - \varepsilon$. 
    \end{itemize}
    The second statement implies that for any $S \subseteq V, w(S) = q$, we have $w(S,S^c) = w(V,V) - w(S,S) - w(S^c,S^c) \leq 1-\Gamma_{\rho}(q) - \Gamma_{\rho}(1-q) + 2\varepsilon$. Therefore, by setting $\gamma$ sufficiently small this shows UG-hardness of approximating CC-Max-Cut with cardinality constraint $q$ within
    \begin{equation*}
      \frac{1-\Gamma_{\rho}(1-q)-\Gamma_{\rho}(q)}{2q(1-q)(1-\rho)} + 2\varepsilon,
    \end{equation*}
    where $\varepsilon >0 $ is arbitrarily small.
    This reduction yields a weighted graph, which can be easily converted into an unweighted multigraph, using e.g.~a simple reduction from Step 1 of Theorem 4.1. in \cite{DBLP:journals/toc/AustrinKS11}. 
  \end{proof}
\end{thm}

\subsection{Hardness for Max-$k$-VC and CC-Max-2-Sat}\label{kvc_2sat_hardness_section}
Next we give the hardness result for Max-$k$-VC.
\begin{thm} \label{max-k-vc}
  Consider $q \in (0,1)$ and let $\rho \in \kappa(q)$. Then, it is UG-hard to approximate Max-$k$-VC with cardinality constraint $q$ within $\beta^{cc}_{vc}(q,\rho)+\varepsilon$ where $\varepsilon >0$ is arbitrary small and $\beta^{cc}_{vc}(q,\rho)$ is given by 
  \begin{equation*}
    \beta^{cc}_{vc}(q,\rho) = \frac{1-\Gamma_{\rho}(1-q)}{q(1+(1-q)(1-\rho))}.
  \end{equation*}
  \begin{proof}
    As we have shown in Theorem \ref{hardness_cc_max_cut}, there is a family of multigraphs $G=(V,E)$ for which it is UG-hard to decide between the following two statements:
    \begin{itemize}
      \item There is a set $S\subseteq V, w(S) = q$, such that $w(S,S^c) \geq 2q(1-q)(1-\rho) -2\gamma$.
      \item For any $r \in [0,1]$ and every set $T\subseteq V, w(T) = r$ we have $w(T,T) \geq \Gamma_{\rho} (r) - \varepsilon$. 
    \end{itemize}
    By \eqref{eqn:weight relation}, the first item implies that $w(S, V) = q(1 + q(1-q)(1-\rho)) - \gamma$.
    The second statement implies that for any $S \subseteq V, w(S) = q$, we have $w(S,V) = w(V,V) -  w(S^c,S^c) \leq 1-\Gamma_{\rho}(1-q) + \varepsilon$. Therefore, by letting $\gamma \to 0 $ this shows UG-hardness of approximating Max-$k$-VC with cardinality constraint $q$ within
    \begin{equation*}
     \frac{1-\Gamma_{\rho}(1-q)}{q(1+(1-q)(1-\rho))} + \varepsilon,
    \end{equation*}
    where $\varepsilon >0 $ is arbitrarily small. As in the CC-Max-Cut case, this reduction yields a weighted graph, which can be converted into an unweighted multigraph by using the reduction from \cite{DBLP:journals/toc/AustrinKS11}. 
  \end{proof}
\end{thm}

\subsection{Hardness as a Function of the Cardinality Contraint}
As we have concluded in Theorems \ref{max-cut-sol} and \ref{max-k-vc}, it is UG-hard to approximate CC-Max-Cut and Max-$k$-VC with cardinality constraint $q \in (0,1)$ to within
\begin{equation*}
  \beta^{cc}_{cut}(q) = \inf_{\rho \in \kappa(q)} \beta^{cc}_{cut}(q,\rho), \quad \beta^{cc}_{vc}(q) = \inf_{\rho \in \kappa(q) } \beta^{cc}_{vc}(q,\rho),
\end{equation*}
respectively. 
In Figure~\ref{hardness_plots_theorem} we visualize these two hardness curves as well as the resulting hardness for CC-Max-$2$-Sat (obtained by taking the lower half of the Max-$k$-VC curve and mirroring it).

\begin{figure}
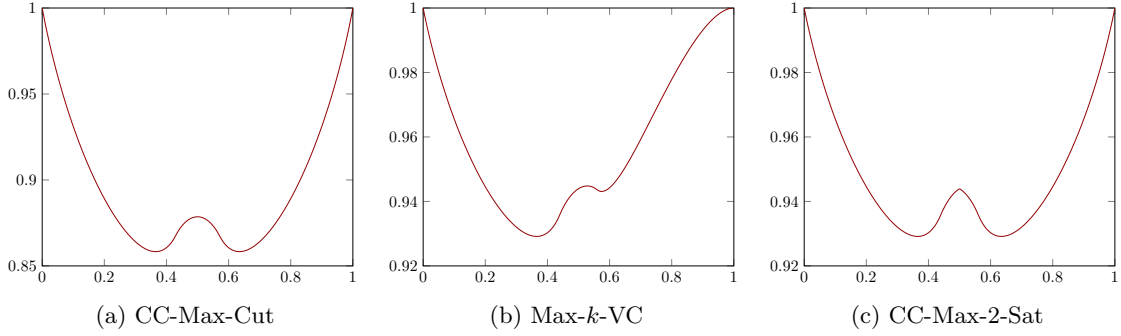

  \begin{subfigure}{0.33\textwidth}
    \resizebox{\linewidth}{!}{
      \input{basic-curve-cc-max-cut}
    }
    \caption{CC-Max-Cut}
    \label{fig basic-curve cc-max-cut}
  \end{subfigure}
  \begin{subfigure}{0.33\textwidth}
    \resizebox{\linewidth}{!}{
      \input{basic-curve-max-k-vc}
    }
    \caption{Max-$k$-VC}
    \label{fig basic-curve max-k-vc}
  \end{subfigure}
  \begin{subfigure}{0.33\textwidth}
    \resizebox{\linewidth}{!}{
      \input{basic-curve-cc-max-2-sat}
    }
    \caption{CC-Max-$2$-Sat}
    \label{fig basic-curve cc-max-2-sat}
  \end{subfigure}
  \caption{Hardness ratio as a function of cardinality constraint $q$ obtained by arguments from Section \ref{cc_max_cut_hardness_section} and Section \ref{kvc_2sat_hardness_section}. }
  \label{hardness_plots_theorem}
\end{figure}

For a fixed $q$ it is not clear for which $\rho$ the functions $\beta^{cc}_{cut}(q,\cdot)$ and $\beta^{cc}_{vc}(q,\cdot)$ are minimized. For the plots of the inapproximability curves in Figure \ref{hardness_plots_theorem}, the optimization over $\rho$ was done numerically. Interestingly, numerical calculations show that the worst-case value of the cardinality constraint $q < 1/2$ (the value of $q$ at which the hardness ratio meets the approximation ratio) is the same for Max-$k$-VC and CC-Max-Cut, and in particular its value is $q^* \approx 0.365$.  
The value of the correlation parameter $\rho$ for which this worst-case hardness is achieved is extremal, i.e., $\rho=-q^*/(1-q^*) \approx -0.575$.
However, the local minima at $q > 1/2$ in the two curves do not occur at the same value of $q$.  For CC-Max-Cut the curve is symmetric around $1/2$ and the minimum occurs at $1-q^* \approx 0.635$, but for the less symmetric Max-$k$-VC problem it occurs at $\approx 0.574$.
\par 
Furthermore, for all $q \le q^*$ and also for $q > 1/2$ greater than the respective local minimum, the $\rho$ minimizing both $\beta^{cc}_{cut}(q,\rho)$ and $\beta^{cc}_{vc}(q,\rho)$ is the minimum value of $\kappa(q)$. On the other hand, when $q$ is close to $1/2$, the best choice of $\rho$ does not equal $\min \kappa(q)$. For example, when $q=1/2$, the hardness we obtain for CC-Max-Cut is the same as for the Max-Cut problem, attained using the value $\rho \approx -0.689$.

\subsection{Improved Hardness Using Isolated Vertices}
\label{section:isolated vertices}
The hardness results shown in Figure~\ref{hardness_plots_theorem} obtained in the preceding sections can be improved by simple monotonicity arguments.  The simplest case is Max-$k$-VC, for which the true approximability curve must be monotone.

\begin{claim}
  \label{claim:max k vc monotone}
  If there is an $\alpha$-approximation algorithm for Max-$k$-VC with
  cardinality constraint $q'$, then there is also an
  $\alpha$-approximation algorithm for Max-$k$-VC for all cardinality
  constraints $q > q'$.
\end{claim}
\begin{proof}
  Given a Max-$k$-VC instance $G$ with cardinality constraint $q$, construct $G'$ by adding $(q/q' - 1)|V|$ isolated vertices to $G$.  Observe that an optimal Max-$k$-VC solution to $G'$ with cardinality constraint $q'$ only uses vertices from $G$ and has cardinality $q'|V'| = q|V|$.
\end{proof}
Applying Claim~\ref{claim:max k vc monotone} to the Max-$k$-VC curve in Figure~\ref{fig basic-curve max-k-vc} we obtain improvements for small $q$ and for $q$ around $1/2$, as shown in Figure~\ref{fig flatten max-k-vc}.

The direct analogue of Claim~\ref{claim:max k vc monotone} is not
obviously true for CC-Max-Cut and CC-Max-2-Sat, because the optimum
value as a function of the cardinality constraint $q$ is not
necessarily monotone (so an optimal solution of $G'$ in the above
reduction might select some of the newly added isolated vertices).
However, note that in the soundness case of our reduction, the
CC-Max-Cut value with cardinality constraint $q$ is (up to $\epsilon$
error) $1 - \Lambda_\rho(q) - \Lambda_\rho(1-q)$, and this is a
monotonically increasing function for $q \le 1/2$.  This implies that
for these instances and $q \le 1/2$ this trivial reduction of adding
isolated vertices still works (and symmetrically for $q \ge 1/2$),
resulting in the improvements shown in
Figure~\ref{fig flatten cc-max-cut} for $q$ bounded away from $1/2$.

Finally, for CC-Max-2-Sat and any value of $q$, we note that by adding
$|V|/\min(q, 1-q)$ new dummy variables, the problem is as hard as unconstrained
Max-2-Sat, giving further small improvements for $q$ close to $1/2$ as
shown in Figure~\ref{fig flatten cc-max-2-sat}.

\begin{figure}
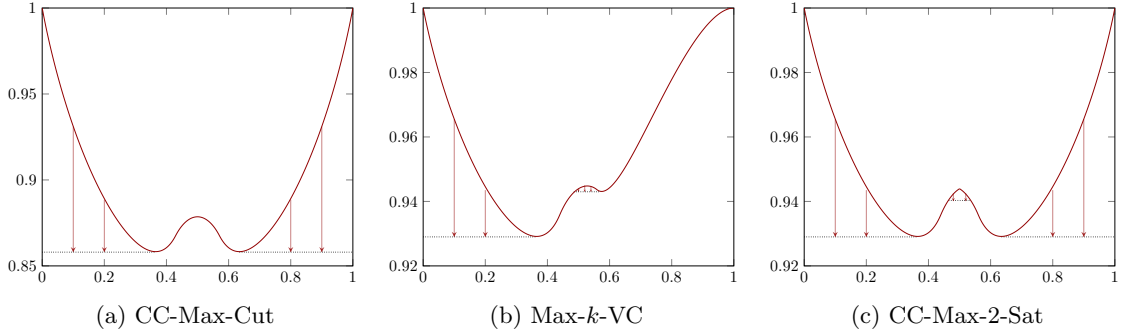

  \begin{subfigure}{0.33\textwidth}
    \resizebox{\linewidth}{!}{
      \input{flatten-cc-max-cut}
    }
    \caption{CC-Max-Cut}
    \label{fig flatten cc-max-cut}
  \end{subfigure}
  \begin{subfigure}{0.33\textwidth}
    \resizebox{\linewidth}{!}{
      \input{flatten-max-k-vc}
    }
    \caption{Max-$k$-VC}
    \label{fig flatten max-k-vc}
  \end{subfigure}
  \begin{subfigure}{0.33\textwidth}
    \resizebox{\linewidth}{!}{
      \input{flatten-cc-max-2-sat}
    }
    \caption{CC-Max-$2$-Sat}
    \label{fig flatten cc-max-2-sat}
  \end{subfigure}
  \caption{Improving hardness using isolated vertices/dummy variables.}
  \label{improved_hardness_plots}
\end{figure}

\section{Approximation Algorithm} \label{RT_section}
In this section we recall the algorithm of Raghavendra and Tan \cite{DBLP:conf/soda/RaghavendraT12}, somewhat reformulated in order to obtain explicit expressions for the approximation ratios that match the hardness results we obtain.  We keep the exposition at a high level and skip over certain technical details, and refer the reader interested in the details to \cite{DBLP:conf/soda/RaghavendraT12} or the follow-up work \cite{DBLP:conf/soda/AustrinBG13}.

\par 
In order to find a good approximation for NP-hard integer optimization problems given in Definitions \ref{Max-Cut_definition} and \ref{max-2-sat-cc_definition} we use semidefinite programming (SDP) relaxations.  
In particular, we extend the domain of variables $\{x_i\}_{i=1}^{n}$ from $\{0,1\}$ to vectors on an $n$-sphere, which we denote by $v_i \in S^n$. We also introduce a vector $v_0 \in S^n$ which represents the value false (corresponding value is $1$ in the integer program). Then, we replace $x_i$ by the scalar product $\langle v_0,v_i \rangle$ and $x_ix_j$ with $\langle v_i,v_j \rangle$. For example, the semidefinite relaxation of the CC-Max-Cut program is given as
\begin{equation*}
  \begin{split}
    \max  &  \sum\limits_{(i,j)=e_{\ell} \in E } \frac{1- \langle v_i,v_j \rangle}{2}, \\
    \textrm{s.t.} & \sum\limits_{i \in V } \langle v_i,v_0 \rangle = r n.
  \end{split}
\end{equation*}

Furthermore, since $|x_i-x_j| \leq |x_i-x_k| + |x_k-x_j|$, we also demand from the vectors $v_i$ to satisfy the triangle inequalities $\| v_i - v_j \|_2^2 \leq \|v_i-v_0\|_2^2 + \|v_0-v_j\|_2^2$.
In order to relax the notation we define $\mu_i = \langle v_0, v_i \rangle$, $\rho_{ij} = \langle v_i, v_j \rangle$, and write triangle inequalities as
\begin{equation*}
  \begin{aligned} \label{triangle_ineq}
   & \mu_i+\mu_j+\rho_{ij}  \geq -1, \quad &  \mu_i - \mu_j-\rho_{ij}  \geq -1,  \\
    - &\mu_i+\mu_j-\rho_{ij} \geq -1, \quad &  -\mu_i - \mu_j+\rho_{ij} \geq -1. 
  \end{aligned}
\end{equation*}
The triples $(\mu_1,\mu_2,\rho)$ satisfying triangle inequalities will be called \emph{configurations}. We denote the set of all configurations as $\mathbf{Conf} \subseteq [-1,1]^3$.
We can solve a semidefinite program up to desired accuracy in polynomial time. Then, the main challenge is finding a \emph{rounding algorithm} which translates the vectors $\{v_i\}_{i=0}^n$ back to $\{-1,1\}$ so that they satisfy the balance constraint, and such that the rounding does not incur a big loss in the objective value. Raghavendra and Tan used a randomized rounding procedure, which rounds vectors $\{v_i\}_{i=0}^n$ to $\pm 1$ integers $\{y_i\}_{i=1}^n$ in the following way. First, let us define $w_i = v_i - \mu_i v_0$, and let\footnote{We assume that $\|w_i\| \neq 0 $, since we can introduce a small perturbation to the values $v_i$ without affecting the objective value too much.} $\overline{w}_i = w_i/\|w_i\|$. Then, we draw a vector $g$ from the Gaussian distribution $\mathcal{N}(0,I^{n+1})$ and set the values of $\overline{y}_i$ as
\begin{equation*}
  \overline{y}_i =
  \left\{
	\begin{array}{ll}
    1  & \mbox{if }  \langle g,\overline{w}_i \rangle  \geq \Phi^{-1}\left(\frac{1-\mu_i}{2}\right), \ \\
		-1 & \mbox{otherwise.}  
	\end{array}
  \right. 
\end{equation*}
It is trivial to check that $\mathbf{E}[ \overline{y}_i]  = \mu_i$, so we have $\mathbf{E}\left [ \sum\limits_{i=1}^{n} \overline{y}_i\right ]  = r n$, and therefore the solution $\{\overline{y}_i\}_{i=1}^{n}$ satisfies the balance constraint in expectation. Furthermore, as shown in \cite{DBLP:conf/soda/RaghavendraT12}, using additional levels of the Lasserre hierarchy we can guarantee that with probability $1-\delta$ the sampled solution $\{\overline{y}_i\}_{i=1}^n$ is $O(\delta)$-far away from satisfying the balance constraint, where $\delta>0$ can be chosen arbitrarily small. Therefore, we can change the values of at most $O(\delta) n$ variables $\overline{y}_i$ to get a solution $y_i$ exactly satisfying the balance constraint, while losing only an additional small factor $O(\delta)$ in the objective value.  Thus, it is sufficient to show that the objective value of the $\overline{y}_i$'s is large.
\par
Consider now the SDP relaxation for any of the integer programs $\mathcal{F}$ given in either Definition \ref{Max-Cut_definition} or Definition \ref{max-2-sat-cc_definition}, and let $\sdpval (\mathcal{F})$ be the optimal value of the SDP relaxation for the instance $\mathcal{F}$. We have that $\sdpval(\mathcal{F}) \geq \optval(\mathcal{F})$. Finally, let us define $\rndval(\mathcal{F})$ to be the expectation of the value of the objective function after randomized rounding procedure.  The analysis of the approximation ratio for the algorithm boils down to proving $\rndval(\mathcal{F}) \geq \alpha\sdpval(\mathcal{F})$, where $\alpha $ is a constant that depends on the problem of interest. The way to calculate $\alpha$ is to look at the loss incurred by rounding at each constraint. Let us now show how this can be done for the CC-Max-Cut problem. \par
The expected value of each constraint $\frac{1-x_i x_j}{2}$ after rounding the SDP solution of CC-Max-Cut problem is $\frac{1-\mathop{\mathbf{E}}\left [  \overline{y}_i \overline{y}_j\right] }{2}$, and therefore at each constraint the loss factor incurred by rounding is given as
\begin{equation*}
  \frac{1-\mathop{\mathbf{E}}\left [  \overline{y}_i \overline{y}_j \right] }{2} \cdot \frac{1}{(1-\langle v_i, v_j \rangle )/2}.
\end{equation*}

Thus, in order to calculate the approximation ratio, we need to bound this expression from below. Let us first note that 
\begin{equation*}
  \mathop{\mathbf{E}}[\overline{y}_1 \overline{y}_2] =  4\Gamma_{\overline{\rho}} \left(\frac{1-\mu_1}{2}, \frac{1-\mu_2}{2} \right) +\mu_1+\mu_2-1,
\end{equation*}
where $\overline{\rho}$ is given as
\begin{equation*}
  \overline{\rho} = \frac{\rho-\mu_1\mu_2}{\sqrt{1-\mu_1^2}\sqrt{1-\mu_2^2}}.
\end{equation*}
Then, the approximation ratio is lower bounded by the quantity $\alpha^{cc}_{cut}$ defined as the solution of the optimization problem
\begin{equation*}
  \alpha^{cc}_{cut} = \min\limits_{(\mu_1,\mu_2,\rho) \in \textbf{Conf}} \frac{2- 4 \Gamma_{\overline{\rho}} \left(\frac{1-\mu_1}{2}, \frac{1-\mu_2}{2} \right) -\mu_1-\mu_2 } { 1-\rho}.
\end{equation*}
Computing $\alpha^{cc}_{cut}$ is a hard global optimization problem, and therefore we resort to numerical computations to estimate it (we remark that the same approach is taken for a similar function in \cite{DBLP:conf/ipco/LewinLZ02} and \cite{DBLP:conf/stoc/Austrin07}). Extensive numerical experiments show that the minimum is attained at $\mu_1=\mu_2=\mu$, while the $\rho$ is on the boundary of the polytope $\mathbf{Conf}$, $\rho = -1+2|\mu|$. More precisely, the minimum is attained at $\mu \approx 0.27  $, and $\rho \approx -0.575$, and it has a value of approximately $0.858$.  \par
Assuming that the minimum is attained at the configuration of the form $(\mu,\mu,-1+2\mu), \mu>0$, constant $\alpha^{cc}_{cut}$ can be found as the minimum of a function
\begin{equation*}
   \frac{1- 2 \Gamma_{\overline{\rho}} \left(\frac{1-\mu}{2}, \frac{1-\mu}{2} \right) -\mu } { 1-\mu}, 
\end{equation*}
where $\mu \in (0,1) $. If we introduce $q = (1-\mu)/2$, we can reexpress this function as
\begin{equation*}
  \alpha^{cc}_{cut}(q)= \frac{2q- 2\Gamma_{\overline{\rho}} \left(q \right)  } {2q } = \frac{1- \Gamma_{\overline{\rho}} \left(q \right)  - \Gamma_{\overline{\rho}} \left(1-q \right)  } {2q } , \quad q \in (0,1/2),
\end{equation*}
where in the last equality we used Lemma \ref{bivariate_lemma}. Furthermore, $\overline{\rho} = -q/(1-q)$. Similar analysis for CC-Max-2-Lin shows that the approximation ratio is the minimal value of the same function. 
\par Straightforward calculations show that $\beta^{cc}_{cut}(q,-q/(1-q))$ from Theorem~\ref{max-cut-sol} equals the value of $\alpha^{cc}_{cut}(q)$. Therefore, under the (mild) assumption that worst-case configurations indeed take the special form as explained above, our hardness result is sharp and the algorithm for CC-Max-Cut of Raghavendra and Tan is optimal on general instances of CC-Max-Cut / CC-Max-2-Lin.

\par 

In completely analogous way, we can conclude that the approximation ratio for CC-Max-2-Sat and Max-$k$-VC problems  can be calculated as the minimum of the following function
\begin{equation*}
  \alpha^{cc}_{2sat}(q) = \frac{1-\Gamma_{\overline{\rho}}(1-q)}{2q},  \quad q \in (0,1/2),
\end{equation*}
where $\overline{\rho} = -q/(1-q)$.  Numerical experiments show that $\alpha^{cc}_{2sat} \approx 0.929$, and that the minimum is attained at $q \approx 0.365$.
\par 
Again we have that the corresponding hardness expression from Theorem~\ref{max-k-vc} satisfies $\beta^{cc}_{vc}(q,-q/(1-q)) = \alpha^{cc}_{2sat}(q)$, implying (under the assumption on worst-case configurations) that the algorithm for CC-Max-2-Sat of Raghavendra and Tan is optimal.
\par
\section{Conclusion and Some Open Questions}\label{discussion_section}

We studied cardinality constrained 2-CSPs, and assuming the Unique Games Conjecture derived hardness results which show that approximation ratios achieved by the algorithm described in \cite{DBLP:conf/soda/RaghavendraT12} are optimal for CC-Max-$2$-Sat (and its special case Max-$k$-VC) and CC-Max-$2$-Lin (and its special case CC-Max-Cut).  An obvious open question is to close the gap between the approximation ratio and hardness for all values of $q$ in Figure~\ref{plot_cc_max_cut}.

It would be interesting to derive UG-hardness for related CC-Max-CSPs of arity $2$, most interestingly for the Max-$k$-DS problem.  While super-constant hardness for Max-$k$-DS is currently known under the closely related Small-Set Expansion Hypothesis \cite{DBLP:conf/stoc/RaghavendraS10}, it is not yet known whether the UGC implies hardness of Max-$k$-DS. Another interesting CSP of arity $2$ is the CC-Max-Di-Cut problem, which as far as we are aware has not been previously studied in the literature.  It has a simple randomized $1/4$-approximation algorithm (pick $k/2$ out of the $k$ vertices with highest out-degree at random, and $k/2$ out of the other $|V|-k$ vertices at random) and is as hard as CC-Max-Cut, but beyond that we are not aware of any results.

Arguably, the simple trick of adding isolated vertices to achieve the flat parts of Figure~\ref{plot_cc_max_cut} in Section~\ref{section:isolated vertices} is somewhat dissatisfactory, and suggests that it may be more interesting to instead study the approximability of these problems on \emph{regular} instances.  The graphs produced by our main reduction are indeed regular so the hardness curves in Figure~\ref{hardness_plots_theorem} still apply for regular graphs.  Furthermore, it is easy to see that in the regular case the best approximation ratio does tend to $1$ as $q$ tends to $0$ or $1$.  For instance, for CC-Max-Cut in a regular graph with cardinality constraint $q \le 1/2$, picking a random set of $q|V|$ vertices gives an approximation ratio of $1-q$, because it is expected to cut a $2q(1-q)$ fraction of all edges, and no cut can cut more than a $2q$ fraction of all edges.  For small $q$ this matches the hardness result in Figure~\ref{fig basic-curve cc-max-cut}, which up to lower order terms equals $1-q$.
\par
Another interesting direction would be to come up with hardness results for cardinality constrained versions of some other well-know Max-CSPs like Max-3-Sat, or even more ambitiously to extend the results of Raghavendra \cite{DBLP:conf/stoc/Raghavendra08} and obtain tight hardness for all cardinality-constrained Max-CSPs.

\section*{Acknowledgements}
The authors thank Johan H\aa stad for helpful suggestions and comments on the manuscript. We also thank anonymous reviewers for their helpful remarks.

\bibliography{bibl}{}
\bibliographystyle{plain}
 
\end{document}